%% file: knapsack.tex
\documentclass[11pt]{article}
\usepackage[left=1in,right=1in,top=1in,bottom=1in]{geometry}
\usepackage{cite}

\usepackage{amsmath}%
\usepackage{amsfonts}%
\usepackage{amssymb}%

\newtheorem{theorem}{Theorem}

\newtheorem{corollary}[theorem]{Corollary}

\newtheorem{definition}[theorem]{Definition}

\newtheorem{lemma}[theorem]{Lemma}

\newenvironment{proof}[1][Proof]{\noindent\textbf{#1.} }{
	\phantom{qed} \hfill \rule{0.5em}{0.5em}
}

\newcommand{\name}[1]{{\mathsf{#1}}}

\newcommand{\intersect}{{\cap}}
\newcommand{\union}{{\cup}}
\newcommand{\set}[1]{{\left\{#1\right\}}}
\newcommand{\powerset}[2][{}]{{{\mathcal{P}_{#1}\left(#2\right)}}}

\newcommand{\st}[1]{{\left|#1\right.}}

\newcommand{\project}[2]{{#1|_{#2}}}
\newcommand{\dotprod}[2]{{\langle#1,#2\rangle}}
\newcommand{\isPSD}{{\succeq 0}}

\newcommand{\lp}[3]{{
  \begin{aligned}
    #1 &\ \ \ #2\\
    \text{s.t.} & \ \ \ \left\{
      \begin{aligned}
        #3
      \end{aligned}
    \right.
  \end{aligned}
}}

\newcommand{\comment}[1]{{}}

\newcommand{\Knapsack}{{\name{Knapsack}}}

\newcommand{\hierarchy}[3]{{\name{#1}^{#2}\left(#3\right)}}
\newcommand{\SA}[2]{{\hierarchy{SA}{#1}{#2}}}
\newcommand{\sa}[2]{{\hierarchy{sa}{#1}{#2}}}
\newcommand{\La}[2]{{\hierarchy{La}{#1}{#2}}}
\newcommand{\la}[2]{{\hierarchy{la}{#1}{#2}}}

\newcommand{\conv}[1]{{\name{conv}\left(#1\right)}}

\title{Integrality Gaps of Linear and Semi-definite Programming Relaxations for Knapsack}
\author{Anna R. Karlin\footnotemark[1]
\and  Claire Mathieu\footnotemark[2] \and C. Thach Nguyen\footnotemark[1]}
\date{}

\begin{document}
\maketitle
\renewcommand{\thefootnote}{\fnsymbol{footnote}}
\footnotetext[1]{University of Washington, email: 
{\tt \{karlin,ncthach\}@cs.washington.edu}. Supported by NSF Grant CCF-0635147 and a Yahoo! Faculty Research Grant.}
\footnotetext[2]{Brown University, email: {\tt claire@cs.brown.edu}}
\renewcommand{\thefootnote}{\arabic{footnote}}

\begin{abstract}
	\input{abstract}

\end{abstract}

\section{Introduction}
\input{introduction}

\subsection{Related Work}
\input{related}

\section{Preliminaries}
\input{definitions}

\section{Lower bound for the Sherali-Adams hierarchy for $\name{Knapsack}$}
\input{sa-lowerbounds}

\section{A decomposition theorem for the Lasserre hierarchy}
\input{la-decomposition}

\section{Upper bound for the Lasserre hierarchy for $\name{Knapsack}$}
\input{la-upperbound}

\section{Conclusion}
\input{conclusion}

\section*{Acknowledgement}
\input{ack}

\bibliography{knapsack}{}
\bibliographystyle{plain}

\section*{Appendix}
\appendix
\input{appendix}

\end{document}

%% file: abstract.tex
In this paper, we study the integrality gap of the $\name{Knapsack}$ 
linear program in
the Sherali-Adams and Lasserre hierarchies.
First, we show that an integrality gap of $2-\epsilon$ persists up to a linear
number of rounds of Sherali-Adams, despite the fact that $\name{Knapsack}$ admits 
a fully polynomial time approximation scheme~\cite{IK,L}.
Second, we show that the Lasserre hierarchy closes the gap quickly.
Specifically, after $t$ rounds of Lasserre, the integrality gap
decreases to $t/(t-1).$
To the best of our knowledge, 
this is the first positive result that uses more than
a small number of rounds in the Lasserre hierarchy.
Our proof uses a decomposition theorem for the Lasserre hierarchy, which 
may be of independent interest.

%% file: introduction.tex
Many approximation algorithms work in two phases: first, solve a linear 
programming (LP) or semi-definite programming (SDP) relaxation; 
then, round the fractional solution to obtain a feasible integer solution 
to the original problem. 
This paradigm is amazingly powerful; in particular, under the unique
game conjecture, it yields the best possible ratio for 
$\name{MaxCut}$ and a wide variety of other problems, 
see e.g.~\cite{Rag08}.

However, these algorithms have a limitation.
Since they are usually analyzed by comparing the value of the output 
to that of the fractional solution, we cannot generally hope to get a better 
approximation ratio than the integrality gap of the relaxation.
Furthermore, for any given combinatorial optimization problem, 
there are many possible LP/SDP relaxations, and it is difficult to determine
which relaxations have the best integrality gaps.

This has lead to efforts to provide systematic procedures for constructing
a sequence of increasingly tight mathematical programming relaxations
for 0-1 optimization problems.
A number of different procedures of this type have been proposed: by Lov\'asz
and Schrijver~\cite{LS}, Sherali and Adams~\cite{SA90}, Balas, Ceria
and Cornuejols ~\cite{BCC}, Lasserre ~\cite{Lass1,Lass2} and others.
While they differ in the details, they all operate in a series of rounds
starting from an LP or SDP relaxation,  eventually ending with an
exact integer formulation. 
The strengthened relaxation after $t$ rounds can typically be solved in 
$n^{O(t)}$ time and, roughly, satisfies the property that the values of 
any $t$ variables in the original relaxation can be expressed as the 
projection of a convex combination of integer solutions.

A major line of research in this area has focused on understanding the
strengths and limitations of these procedures.  Of particular interest
to our community is the question of how the integrality gaps for
interesting combinatorial optimization problems evolve through a
series of rounds of one of these procedures. On the one hand, if the
integrality gaps of successive relaxations drop sufficiently fast,
there is the potential for an improved approximation algorithm (see \cite{Chlamtac, CS, BCG,BO} for example). 
On the other hand, a large integrality gap persisting for
a large, say logarithmic, number of rounds rules out (unconditionally) a very
wide class of efficient approximation algorithms, namely those whose output
is analyzed by comparing it to the value of a class of LP/SDP relaxations.
This implicitly contains
most known sophisticated approximation algorithms for many problems including
$\name{Sparsest Cut}$ and $\name{Maximum Satisifiability}$.
Indeed, serveral very strong negative results of this type have been 
obtained (see~\cite{ABLT,AAT,BGHMP,CMM,GMPT,PS,Schoen,STT,Tourlakis,RS,GMT,BM} 
and others).
These are also viewed as lower bounds of approximability in certain 
restricted models of computation.

How strong are these restricted models of computation? 
In other words, how much do lower bounds in these models tell us about
the intrinsic hardness of the problems studied?
To explore this question, we focus on one problem that is well-known
to be ``easy'' from the viewpoint of approximability: $\name{Knapsack}$.
We obtain the following results:

\begin{itemize}

\item We show that an integrality gap close to 2 persists up to
a linear number of rounds of Sherali-Adams. (The integrality 
gap of the natural LP is 2.) 

This is interesting since $\name{Knapsack}$ has a fully
polynomial time approximation scheme~\cite{IK,L}. 
This confirms and amplifies what has already been observed in other contexts 
(e.g ~\cite{CMM}): the Sherali-Adams restricted model of computation has 
serious weaknesses: a lower bound in this model does not necessarily 
imply that it is difficult to get a good approximation algorithm.

\item We show that Lasserre's hierarchy closes the gap
  quickly. Specifically, after $t$ rounds of Laserre, the
  integrality gap decreases to $t/(t-1)$. 

It is known that a few rounds of Lasserre can yield better relaxations.
For example, two rounds of Lasserre applied to the $\name{MaxCut}$ LP yields
an SDP that is at least as strong as that used by 
Goemans and Williamson to get the 
best known approximation algorithm, and the SDP in~\cite{ARV09} which
leads to the best known approximation algorithm for $\name{Sparsest Cut}$ can be 
obtained by three rounds of Lasserre. 
However, to the best of our knowledge,
this is the first {positive} result that utilizes more than a small constant
number of rounds in the Lasserre hierarchy.

 \end{itemize}

%% file: related.tex
Many known approximation algorithms can be recognized in hindsight as 
starting from a natural relaxation and strengthening it using a couple 
of levels of lift-and-project. 
The original hope~\cite{ABL} had been to use lift and project systems 
as a systematic approach to designing novel algorithms with better 
approximation ratios. 
Instead, the last few years have mostly 
seen the emergence of a multitude of lower bounds. 
Indeed, lift and project systems have been studied mostly for well known 
difficult problems: $\name{MaxCut}$ \cite{CMM,FdlVKM,STT}, 
$\name{Sparsest Cut}$, 
\cite{CMM,CKN} $\name{Vertex Cover}$ 
\cite{AAT,ABL,ABLT,Charikar,GMPT,GMT,HMM,STT,Tourlakis}, 
$\name{Hypergraph Vertex Cover}$, $\name{TSP}$ \cite{C2}, 
$\name{Maximum Acyclic Subgraph}$ \cite{CMM}, $\name{CSP}$ \cite{Schoen,T}, 
and more.

The $\name{Knapsack}$ problem~\cite{MT,KPP}  has a  fully
polynomial time approximation scheme~\cite{IK,L}. 
The natural LP relaxation (to be stated in full detail in the next section) 
has an integrality gap of $2-\epsilon$ \cite{KPP}. 
Although we are not aware of previous work on using the lift and project 
systems for $\name{Knapsack}$, the problem of strengthening the LP relaxation via 
addition of well-chosen inequalities has been much the object of much interest 
in the past in the mathematical programming community, as 
stronger LP relaxations are extremely useful to speed up branch-and-bound 
heuristics. 
The knapsack polytope was studied in detail by Weismantel \cite{W}. 
Valid inequalities were studied in \cite{B,HJP,HP,W2,BZ}. 
In particular, whenever $S$ is a minimal set (w.r.to inclusion) that does not 
fit in the knapsack, then 
$\sum_{S\cup \{j: \forall i\in S, w_j\geq w_i\} } x_j\leq |S|-1$ is a 
valid inequality. 
Generalizations and variations were also studied in \cite{EGP,HZ,Z}.
Thus, in spite of the existence of a dynamic program to solve the 
problem, $\name{Knapsack}$ is fundamental enough that understanding the 
polytope (and its lifted tightenings) is of intrinsic interest.
In \cite{Bie08}, Bienstock formulated LP
with arbitrary small integrality gaps for $\name{Knapsack}$ using 
``structural disjunctions'', and asked if the popular hierarchies reduce
the gap of the $\name{Knapsack}$ linear program. 
Our results give a negative answer for Sherali-Adams and a strong
affirmative one for Lasserre.

Our results confirm the indication from \cite{KS,RS} for example that the 
Sherali-Adams lift and project is not powerful enough to be an indicator
of the hardness of problems.
However, it should be noted that if the problem was phrased as a decision problem and the objective function was replaced by an additional constraint of the constraint polytope, then Sherali-Adams would succeed in reducing the integrality gap; thus the choice of the initial LP formulation is critical.
On the other hand, little is know about the Lasserre hierarchy, as 
the first negative results were about $k$-CSP \cite{Schoen,T}. 
Our positive result leaves open the possibility that the Lasserre hierarchy 
may have promise as a tool to capture the intrinsic difficulty of problems.

%% file: definitions.tex
\subsection{The $\Knapsack$ problem}
Our focus in this paper is on the
$\Knapsack$ problem. In the $\Knapsack$ problem,  we are given a set of $n$ objects $V=[n]$
with sizes $c_1, c_2, \ldots c_n$,  values $v_1, v_2, \ldots v_n$, and
a capacity $C$.  We assume that for every $i$, $c_i\leq C$.
The objective is to select a subset of objects of maximum total value
such that the total size of the objects selected does not exceed $C$.

The standard linear programming (LP) relaxation~\cite{KPP} for 
$\Knapsack$ is given by:
\begin{align}
\lp{\max}{\sum_{i \in V} v_ix_i}{
		\begin{aligned}
			&\sum_{i \in V}c_ix_i \leq C & \\
			&0 \leq x_i \leq 1 & \ \ \ \forall i \in V
		\end{aligned}
	}
	\label{eq:knapsack-lp}
\end{align}
The intended intepretation of an integral solution of this LP
is obvious: $x_i = 1$ means the object $i$ is selected, and $x_i = 0$ 
means it is not. The constraint can be written as $g(x)=C-\sum_i c_ix_i\geq 0$.

Let {\em Greedy} denote the algorithm that puts objects in the knapsack by order of decreasing ratio $v_i/c_i$, stopping as soon as the next object would exceed the capacity.
The following lemma is folklore.

\begin{lemma}\label{ref:known}
  Consider an instance $(C,V)$ of $\Knapsack$ and its LP relaxation $K$ given by (\ref{eq:knapsack-lp}).  Then
  $$\text{Value}(K)\leq\text{Value(Greedy}(C,V))+\max_{i\in V}v_i.$$
\end{lemma}

\subsection{The Sherali-Adams and Lasserre hierarchies}

We next review the lift-and-project hierarchies that we will use in this paper.
The descriptions we give here  assume that the base program is linear and mostly use
the notation given in the survey paper by Laurent ~\cite{L01}.
To see that these hierarchies apply at a much greater level of generality we refer
the reader to Laurent's paper ~\cite{L01}.

Let $K$ be a polytope defined by a set of linear constraints 
$g_1, g_2, \ldots g_m$:
\begin{equation}
\label{eq:defK}
K = \set{x \in [0,1]^n \st{g_\ell(x) \geq 0\ \text{for}\ \ell = 1,2,\ldots m}}.
\end{equation}
We are interested in optimizing a linear objective function $f$ over 
the convex hull 
$P = \conv{K \intersect\set{0,1}^n}$
of integral points in $K$. 
Here, $P$ is the set of convex combinations of all integral solutions of the
given combinatorial problem and $K$ is the set of solutions to its linear relaxation.
For example, if $K$ is defined by~\eqref{eq:knapsack-lp}, then $P$ is the
set of convex combinations of valid integer solutions to $\Knapsack$.

If all vertices of $K$ are integral then $P = K$ and we are done.
Otherwise, we would like to strengthen the relaxation $K$ by adding 
additional valid constraints.
The Sherali-Adams (SA) and Lasserre hierarchies are two 
different systematic ways to construct these
additional constraints.
In the SA hierarchy, all the constraints added are linear, whereas
Lasserre's hierarchy is stronger and introduces a set of positive
semi-definite constraints. 
However, for consistency, we will describe both hierarchies as 
requiring certain submatrices to be positive semi-definite (readers who are
not familiar with the following formulation of SA are referred to 
Appendix~\ref{a:sa} for a linear formulation of the hierarchy.)

To this end, we first state some notation.
Throughout this paper we will use $\powerset{V}$ to denote the power set of $V$,
and $\powerset[t]{V}$ to denote  the collection of
all subsets of $V$ whose sizes are at most $t$.
Also, given  two sets of coordinates $T$ and $S$, $T \subseteq S$ and $y \in R^{S}$,
by $\project{y}{T}$ we denote the projection of $y$ onto $T$.

Next, we review the definition of the {\em shift operator} 
between two vectors 
$x,y \in R^{\powerset{V}}$: $x*y$ is a vector in 
$R^{\powerset{V}}$ such that
\begin{align*}
	(x*y)_I = \sum_{J \subseteq V} x_Jy_{I\union J}.
\end{align*}
\begin{lemma}[~\cite{L01} ]The shift operator is commutative:
for any vectors $x,y,z \in R^{\powerset{V}}$, we have $x*(y*z) = y*(x*z)$.
\label{lemma:commutative}
\end{lemma}

A polynomial $P(x) = \sum_{I \subseteq V} a_I\prod_{i\in I}x_i$ 
can also be viewed as a vector indexed by subsets of $V$. We define the vector
$P*y$ accordingly:
$(P*y)_I = \sum_{J \subseteq V} a_Jy_{I \union J}.$

Finally, let $\cal T$ 
be a collection of subsets of $V$ and $y$ be a vector in $R^{\cal T}$. 
We denote by $M_{\cal T}(y)$ the matrix whose rows and colums
are indexed by elements of $\cal T$ such that 
\begin{align*}
	\left(M_{\cal T}(y)\right)_{I,J} = y_{I \union J}.
\end{align*}

The main observation is that if $x \in {K \intersect\set{0,1}^n}$ then 
$(y_I )= (\prod_{i \in I}x_i)$ satisfies
$M_{\powerset{V}}(y) = yy^T \isPSD$ and 
$M_{\powerset{V}}(g_\ell* y) = g_\ell(x)yy^T \isPSD$ for all constraints
$g_\ell$. Thus requiring principal submatrices of these two matrices
to be positive semi-definite yields a relaxation. 

\begin{definition}
	\label{def:SA-psd}
  For any $1 \leq t \leq n$, the \emph{$t$-th Sherali-Adams lifted polytope}
  $\SA{t}{K}$ is the set of vectors $y \in [0,1]^{\powerset[t]{V}}$ such
	that $y_\emptyset = 1$, $M_{\powerset{U}}(y) \isPSD$ and 
	$M_{\powerset{W}}(g_\ell*y) \isPSD$
	for all $\ell$ and subsets $U,W \subseteq V$ such that
	$|U| \leq t$ and $|W| \leq t-1$. 
	
	We say that a point $x \in [0,1]^n$ belongs to the 
	\emph{$t$-th Sherali-Adams
  polytope $\sa{t}{K}$} iff there exists a $y \in \SA{t}{K}$ such that
  $y_\set{i} = x_i$ for all $i \in [n]$.
\end{definition}

\begin{definition}
	\label{def:La}
	For any $1 \leq t \leq n$, the \emph{$t$-th Lasserre lifted polytope}
  $\La{t}{K}$ is the set of vectors $y \in [0,1]^{\powerset[2t]{V}}$ such
	that $y_\emptyset = 1$, $M_{\powerset[t]{V}}(y) \isPSD$ and 
	$M_{\powerset[t-1]{V}}(g_\ell*y) \isPSD$ for all $\ell$.
	
	We say that a point $x \in [0,1]^n$ belongs to the \emph{$t$-th Lasserre
  polytope $\la{t}{K}$} if there exists a $y \in \La{t}{K}$ such that
  $y_\set{i} = x_i$ for all $i \in V$.
\end{definition}

Note that $M_{\powerset{U}}(y)$ has at most $2^t$ rows and columns, which is constant for $t$ constant, whereas $M_{\powerset[t]{V}}(y)$ has ${n+1\choose t+1}$ rows and columns.

It is immediate from the definitions that $\sa{t+1}{K} \subseteq \sa{t}{K}$,
and $\la{t+1}{K} \subseteq \la{t}{K}$ for all $1 \leq t \leq n-1$.
Sherali and Adams~\cite{SA90} show that $\sa{n}{K} = P$,
and Lasserre~\cite{Lass1,Lass2} show that $\la{n}{K} = P$.
Thus, the sequences
\begin{align*}
	K \supseteq \sa{1}{K} \supseteq \sa{2}{K} \supseteq \cdots 
	\supseteq \sa{n}{K} = P \\
	K \supseteq \la{1}{K} \supseteq \la{2}{K} \supseteq \cdots 
	\supseteq \la{n}{K} = P 
\end{align*}
define hierarchies of polytopes that converge to $P$.
Furthermore, the Lasserre hierarchy is stronger than the Sherali-Adams
hierarchy: $\la{n}{K} \subseteq \sa{n}{K}$. 
In this paper, we show that for the $\Knapsack$ problem,
the Lasserre hierarchy is strictly stronger.

\comment{
\subsection{The integrality gap of the hierarchies}

We will be studying the integrality gaps of the lifted polytopes for 
$\Knapsack$. For a given $\Knapsack$ instance
$K$, a level $t$ and a particular hierarchy $H$ (where $H$ is either 
Sherali-Adams or Lasserre),
define $H_t(K)$ to be the ratio between the optimal solution in the 
$t$-th lifted
polytope (the maximum value of $\sum_i v_i y_\set{i}$, for
$y$  in the $t$-th lifted polytope on instance $K$) 
and the optimal value of the objective function for a feasible 
integral solution.
The {\em integrality gap at the $t$-th level of the hierarchy $H$} is
$\sup_{K} H_t ({K})$.
}

%% file: sa-lowerbounds.tex
In this section, we show that the integrality gap of the $t$-th level
of the Sherali-Adams hierarchy for $\Knapsack$ is close to 
$2$.  This lower bound even holds for the 
\emph{uniform} $\Knapsack$ problem, in which $v_i=c_i = 1$ for all $i$
\footnote{Some people call this problem \emph{Unweighted Knapsack} or
\emph{Subset Sum.}}.

\begin{theorem}
\label{thm:sa-lowerbound}
For every $\epsilon, \delta > 0$,
the integrality gap at the $t$-th level of the 
Sherali-Adams hierarchy for $\Knapsack$ where $t \leq \delta n$ is at least 
$(2 - \epsilon)(1/(1+\delta))$.
\end{theorem}

\begin{proof} (Sketch - for full proof see Appendix~\ref{appendix:SAproof}.)
	Consider the instance $K$ of $\Knapsack$ with $n$ objects where
	$c_i = v_i = 1$ for all $i \in V$ and capacity $C = 2(1-\epsilon)$.
	Then the optimal integer value is 1. 
	On the other hand, we claim that the vector $y$ where
	$y_\emptyset = 1$, $y_\set{i} = C/(n+(t-1)(1-\epsilon))$ 
	and $y_I = 0$ for all $|I| > 1$
	is in $\SA{t}{K}$.
	Thus, the integrality gap of the $t$th round of Sherali-Adams is at least
	$Cn/(n+(t-1)(1-\epsilon))$, which is at least $(2-\epsilon)(1/(1+\delta))$
	when $t \leq \delta n.$
\end{proof}

%% file: la-decomposition.tex
In this section, we develop the machinery we will need for our Lasserre upper bounds. It turns out that it is more convenient to work with families $(z^X)$ of characteristic vectors rather than directly with $y$. We begin 
with some definitions and basic properties.

\begin{definition}[extension]
\label{def:y'}
Let ${\cal T}$ be a collection of subsets of $V$
and let $y$ be a vector indexed by sets of $\cal T$. We define the \emph{extension} of $y$ 
to be the vector $y'$, indexed by all subsets of $V$,  such that $y'_I$ equals $y_I$ if $I\in\cal T$ and equals $0$ otherwise.
\end{definition}

\begin{definition}[characteristic polynomial]
\label{def:zX}
Let $S$ be a subset of $V$ and $X$ a subset of $S$.
We define the \emph{characteristic polynomial} $P^X$ of
$X$ with respect to $S$ as
$$  P^X(x) = \prod_{i \in X}x_i\prod_{j \in S\backslash X}(1-x_j) =\sum_{J: X\subseteq J\subseteq S}
(-1)^{|J\setminus X|} \prod_{i\in J} x_i.
$$
\end{definition}

\begin{lemma}[inversion formula]
  \label{lem:decomp1}
Let $y'$ be a vector indexed by all subsets of $V$.  Let $S$ be a subset of $V$ and, for each $X$ subset of $S$, let $z^X=P^X * y' $:
$$z^X_I=\sum_{J: X\subseteq J\subseteq S} (-1)^{|J\setminus X|} y'_{I\cup J}.$$ Then
  $y' = \sum_{X \subseteq S}z^X$.
\end{lemma}
\begin{proof} 
Fix a subset $I$ of $V$. Substituting the definition of $z^X_I$ in $\sum_{X \subseteq S}z^X_I $, and changing the index of summation, we get
$$\sum_{X \subseteq S}z^X_I = \sum_{A\subseteq S}\sum_{J\subseteq A} (-1)^{|J|}~y'_{I\cup A}.$$
For $A\neq\emptyset$ the inner sum is 0, so only the term for $A=\emptyset$, which equals $y'_I$, remains. 
\end{proof}

\begin{lemma}
  \label{lem:w}
Let $y'$ be a vector indexed by all subsets of $V$, $S$ be a subset of $V$ and $X$ be a subset of $S$. Then
\begin{equation*}
 \left\{\begin{aligned}
           &z^X_I = z^X_{I \backslash X}&\text{for all }I\\
             &z^X_I=z^X_{\emptyset} &\text{if}\ I \subseteq X\\
            &z^X_I= 0  &\text{if $I \cap (S\setminus X) \ne \emptyset$}
         \end{aligned}\right.
\end{equation*}
\end{lemma}
\begin{proof}
Let $I' = I \setminus X$ and $I'' = I \intersect X$. Using the definition of $z^X_I$ and noticing that $X\cup I''=X$ yields $z^X_I=z^X_{I'}$. This immediately implies that for $I\subseteq X$, $z^X_I=z^X_{\emptyset}$. 

Finally, consider a set $I$ that intersects $S\setminus X$ and let  $i \in I \intersect (S\backslash X)$. In the definition of $z^X_I$, we group the terms of the sum into pairs consisting of $J$ such that $i\notin J$ and of $J\cup \{ i\}$. Since $I=I\cup \{ i\}$, we obtain:
$$	\sum_{J: X\subseteq J\subseteq S} (-1)^{|J\setminus X|} y'_{I\cup J}=
	\sum_{J: X\subseteq J\subseteq S\setminus \{ i\} }
	\left(      (-1)^{|J\setminus X|}+(-1)^{|J\setminus X|+1} \right )  y'_{I\cup J} =0.$$
\end{proof}

\begin{corollary}
  \label{lem:w-integer}
Let $y'$ be a vector indexed by all subsets of $V$, $S$ be a subset of $V$ and $X$ be a subset of $S$. Let $w^X$ be defined as $z^X/z^X_{\emptyset}$ if $z^X_\emptyset\neq 0$ and defined as 0 otherwise. Then, if  $z^X_\emptyset\neq 0$, then $w^X_{\{ i\} }$ equals 1 for elements of $X$ and 0 for elements of $S\setminus X$.
\end{corollary}

\medskip

\begin{definition}[closed under shifting]
Let $S$ be an arbitrary subset of $V$ and $\cal T$ be a collection of subsets of $V$. 
We say that $\cal T$ is \emph{closed under shifting by S} if
$$Y\in{\cal T}\implies \forall X\subseteq S,~~X\cup Y\in {\cal T}.$$
\end{definition}
The following lemma generalizes Lemma 5 in
\cite{L01}. It proves that the positive-semidefinite property carries over from $y$ to $(z^X)$.

\begin{lemma}
  \label{lem:star-psd}
Let $S$ be an arbitrary subset of $V$ and $\cal T$ be a collection of subsets of $V$ that is closed under shifting by $S$. Let $y$ be a vector indexed by sets of $\cal T$. Then
$$M_{\cal T}(y) \isPSD \implies \forall X\subseteq S,~~ M_{\cal T}(z^X) \isPSD .$$
 \end{lemma}
\begin{proof}
	Since $M_{\cal T}(y) \isPSD$, there exist vectors $v_{I}$, 
	$I \in {\cal T}$, such that $\dotprod{v_{I}}{v_{J}} = y_{I \union J}$. Fix a subset $X$ of $S$.
	For each $I \in {\cal T}$, let
	\begin{equation*}
		w_I = \sum_{H \subseteq S\setminus X}(-1)^{|H|} v_{I \union X \union H},
	\end{equation*}
	which is well-defined since $\cal T$ is closed under shifting by $S$.

	Let $I,J\in \cal T$. It is easy to check  that $\dotprod{w_I}{w_J} = (z^X)_{I \union J}$. Indeed,
	\begin{align}
			\dotprod{w_I}{w_J} 
				&= \sum_{H \subseteq S\setminus X}\sum_{L \subseteq S\setminus X} (-1)^{|H| + |L|}
							\dotprod{v_{I\union X \union H}}
											{v_{J\union X \union L}}\\
				&= \sum_{H \subseteq S\setminus X}\sum_{\L \subseteq S\setminus X} (-1)^{|H| + |L|}
							y_{I \union J \union X \union H \union L}
							\label{eq:shift}
	\end{align}
	by definition of $v_I,v_J$ and since $\cal T$ is closed under shifting by $S$ (so that this is well-defined). Consider a non-empty subset $H$ of $S\setminus X$ and let $i\in H$. We group the terms of the inner sum into pairs consisting of $L$ such that $i\notin L$ and of $L\cup \{ i\}$. Since $H=H\cup \{ i\}$, we obtain:
$$		\sum_{L \subseteq S\setminus X}(-1)^{|H| + |L|}
				y_{I \union J \union X \union H \union L}
 = \sum_{L \subseteq (S\setminus X\backslash\set{i})} \left(
						(-1)^{|H| + |L|}
						+ (-1)^{|H| + |L| + 1}
					\right)
						y_{I \union J \union X \union H \union L}=0.$$
	Thus, the expression in~\eqref{eq:shift} becomes
$$\dotprod{w_I}{w_J} =			\sum_{L\subseteq S\setminus X} (-1)^{|L|}
				y_{I \union J \union X \union L}
			 = (z^X)_{I \union J}.
$$
This implies that $M_{\cal T}(z^X) \isPSD$.
\end{proof}

In the rest of the  section, we prove a decomposition theorem for the Lasserre hierarchy,
which allows us to ``divide'' the action of the hierarchy and think of it as
using the first few rounds on some subset of variables, and the other rounds
on the rest.
We will use this theorem to prove that the Lasserre hierarchy 
closes the gap for the $\Knapsack$ problem in the next section.

\begin{theorem}
  \label{thm:la-decomp}
Let $t>1$ and $y \in \La{t}{K}$. Let $k<t$ and $S$ be a subset of $V$ and  such that 
\begin{equation}  |I \intersect S| \geq k \implies  y_I = 0.
\label{eq:assumption}
\end{equation}
Consider the projection $\project{y}{\powerset[2t-2k]{V}}$ of $y$ to the coordinates corresponding to subsets of size
	at most $2t-2k$ of $V$.
	Then there exist subsets $X_1,X_2,\ldots , X_m$  of $S$ such that $\project{y}{\powerset[2t-2k]{V}}$ is a convex combination of vectors $w^{X_i}$ with the following properties:
	\begin{itemize}
	\item   $w^{X_i}_{\{ j \} }=\left\{ \begin{array}{ll}
	1& \text{ if }j\in X_i\\ 0 &\text{ if }j\in S\setminus X_i; \end{array}\right. $ 
	\item $w^{X_i}\in \La{t-k}{K}$; and
	\item  if $K_i$ is obtained from $K$ by setting  $x_j=w^{X_i}_{\{ j\} }$ for $j\in S$, then
  $\project{w^{X_i}}{\powerset[2t-2k]{V\backslash S}} \in \La{t-k}{K_i}$. 
	\end{itemize}
\end{theorem}

To prove Theorem~\ref{thm:la-decomp}, we will need a couple more lemmas. In the first one, using assumption~(\ref{eq:assumption}), we extend the positive semi-definite properties from $y$ to $y'$, and then, using Lemma~\ref{lem:star-psd}, from $y'$ to $z^X$.
\begin{lemma}
  \label{lem:extended-matrix}
  Let $t,y,S, k$ be defined as in Theorem~\ref{thm:la-decomp}, and $y'$ be the extension of $y$. Let ${\cal T}_1 = \set{A \text{ such that } {|A \backslash S| \leq t-k}}$, and
  ${\cal T}_2 = \set{B \text{ such that } {|B \backslash S| < t-k}}$.
  Then  for all $X \subseteq S$,  $M_{{\cal T}_1}(z^X)\isPSD$
   and, for all $\ell$, $M_{{\cal T}_2}(g_\ell * z^X)\isPSD$ .
   \end{lemma}
\begin{proof} We will first prove that $M_{{\cal T}_1}(y')\isPSD$ and, for all $\ell$, $M_{{\cal T}_2}(g_\ell * y')\isPSD$.  
 Order the columns and rows of $M_{{\cal T}_1}(y')$  by subsets of non-decreasing size.
By definition of ${\cal T}_1$, any $I\in{\cal T}_1$ of size at least $t$ must have   $|I \intersect S| \geq k$, and so $y'_I=0$. Thus
$$    M_{{\cal T}_1}(y') = 
      \begin{pmatrix}
        M & 0 \\
        0 & 0 
      \end{pmatrix},$$
      where $M$ is a principal submatrix of $M_{\powerset[t]{V}}(y)$.Thus $M\isPSD$, and so $M_{{\cal T}_1}(y')\isPSD$.
      
Similarly, any $J\in{\cal T}_2$ of size at least $t-1$ must have $|J\cup \{ i\} \intersect S| \geq k$ for every $i$ as well as $|J\intersect S|\geq k$, and so, by definition of $g_\ell * y'$ we must have $(g_\ell * y')_J = 0$. Thus 
$$    M_{{\cal T}_2}(g_\ell * y') =
      \begin{pmatrix}
        N & 0 \\
        0 & 0 
      \end{pmatrix},
 $$ where  $N$ is a 
  principal submatrix of $M_{\powerset[t-1]{V}}(g_\ell * y)$.
  Thus $N\isPSD$, and so  $M_{{\cal T}_2}(g_\ell * y')\isPSD$.
  
  Observe that ${\cal T}_1$ is closed under shifting 
  by $S$. By definition of $z^X$ and Lemma~\ref{lem:star-psd}, we thus get $M_{{\cal T}_1}(z^X)\isPSD$.

  Similarly, observe that ${\cal T}_2$ is also closed under shifting 
  by $S$.  By Lemma~\ref{lemma:commutative}, we have $g_\ell* (P^X* y')=P^X* (g_\ell * y')$, and so
  by Lemma~\ref{lem:star-psd} again we get $M_{{\cal T}_2}(g_\ell* z^X)\isPSD$.
\end{proof}

\begin{lemma}
  \label{lem:decomp2}
  Let $t,y,S, k$ be defined as in Theorem~\ref{thm:la-decomp}, and $y'$ be the extension of $y$. 
  Then for any $X \subseteq S$:
\begin{enumerate}
\item $z^X_{\emptyset} \ge 0$.
\item If $z^X_{\emptyset} = 0$ then $z^X_I = 0$ for all
  $|I| \leq 2t - 2k$.
\end{enumerate}
\end{lemma}
\begin{proof}
Let ${\cal T}_1$ be defined as in Lemma~\ref{lem:extended-matrix}.
By Lemma~\ref{lem:extended-matrix}
  $   M_{{\cal T}_1}(z^X)\isPSD$
and 
$z^X_{\emptyset}$ is a diagonal element of this matrix, hence $z^X_{\emptyset} \ge 0$.

For the second part,
  start by considering  $J\subseteq V$ of size at most $ t-k$. Then 
  $J \in {\cal T}_1$, and so 
  the matrix $M_\set{\emptyset, J}(z^X)$ is  a principal submatrix of 
  $M_{{\cal T}_1}(z^X)$, hence is also positive semidefinite. Since $z^X_\emptyset =0$,
   \begin{equation*}
    M_\set{\emptyset, J}(z^X) =
     \begin{pmatrix}
        0 & z^X_J \\
        z^X_J & z^X_J
      \end{pmatrix}\isPSD ,
  \end{equation*}
  hence $z^X_J = 0.$

  Now consider any $I \subseteq V$ such that $|I| \leq 2t-2k$, and write
  $I = I_1 \union I_2$ where $|I_1| \leq t-k$ and $|I_2| \leq t-k$.
  $M_\set{I_1, I_2}(z^X)$ is  a principal submatrix of 
  $M_{{\cal T}_1}(z^X)$, hence is also positive semidefinite. Since $z^X_{I_1}=z^X_{I_2}=0$,
  Since 
  \begin{equation*}
    M_\set{I_1, I_2}(z^X) =
     \begin{pmatrix}
        0 & z^X_I \\
        z^X_I & 0
      \end{pmatrix}\isPSD ,
  \end{equation*}
 hence $z^X_I = 0.$
\end{proof}

We now have what we need to prove Theorem~\ref{thm:la-decomp}.
\vspace{.5cm}

\begin{proof}[Proof of Theorem~\ref{thm:la-decomp}]
By definition,  Lemma~\ref{lem:decomp1} and the second part of Lemma~\ref{lem:decomp2}, we have
  \begin{equation*}
    \project{y}{\powerset[2t-2k]{V}} =
      \project{y'}{\powerset[2t-2k]{V}} =
      \sum_{X \subseteq S}\project{z^X}{\powerset[2t-2k]{V}} =
      \sum_{X \subseteq S}{z^X_\emptyset}\project{w^X}{\powerset[2t-2k]{V}}.
  \end{equation*}
By Lemma~\ref{lem:decomp1} and by definition of $y$, we have $\sum_{X \subseteq S}{z^X_\emptyset}=y_{\emptyset}=1$, and the terms are non-negative by  the first part of
Lemma~\ref{lem:decomp2}, so
$\project{y}{\powerset[2t-2k]{V}}$ is a convex combination of $w^X$'s, as desired.


Consider $X \subseteq S$ 
  such that $z^X_\emptyset \neq 0$.
  By Lemma~\ref{lem:extended-matrix}, $M_{{\cal T}_1}(z^X)\isPSD$ and $M_{{\cal T}_2}(g_\ell * z^X)\isPSD$ for all $\ell$, and so this also holds for their principal submatrices  $M_{\powerset[t-k]{V}}({z^X}) $ and 
  $M_{\powerset[t-k-1]{V}}(g_\ell * z^X)$. Scaling by the positive quantity $z^X_\emptyset $, by definition of $w^X$ this also holds for $M_{\powerset[t-k]{V}}({w^X})$ and 
  $M_{\powerset[t-k-1]{V}}(g_\ell * w^X)$. In other words,  $\project{w^X}{\powerset[2t-2k]{V}} \in \La{t-k}{K}$.

Since 
$M_{\powerset[t-k]{V}}(w^{X_i})\isPSD$, by  taking a principal submatrix, we infer $M_{\powerset[t-k]{V\backslash S}}(w^{X_i})\isPSD$.
Similarly, $M_{\powerset[t-k]{V}}(g_\ell*w^{X_i})\isPSD$ and so $M_{\powerset[t-k]{V\backslash S}}(g_\ell*w^{X_i})\isPSD$. 
Let $g'_\ell$ be the constraint of $K_i$ obtained from $g_\ell$ by setting
$x_j = w^{X_i}_\set{j}$ for all $j \in S$.
We claim that for any $I \subseteq V \backslash S$, 
$(g'_\ell * z^{X_i} )_I = (g_\ell * z^{X_i})_I$;
scaling implies that $M_{\powerset[t-k]{V\backslash S}}(g'_\ell*w^{X_i})=M_{\powerset[t-k]{V\backslash S}}(g_\ell*w^{X_i})$ and we are done. 

To prove the claim, let $g_\ell(x) = \sum_{j \in V}a_jx_j + b$.
Then, by Corollary~\ref{lem:w-integer}, $g'_\ell = \sum_{j \in V\backslash S}a_jx_j + (b + \sum_{j \in X_i}a_j)$.
Let $I \subseteq V\backslash S$. We see that
$$	(g_\ell*w^{X_i})_I -(g'_\ell*w^{X_i})_I
		= \sum_{j \in X_i}a_j	w^{X_i}_{I\union\set{j}} 
		+ \sum_{j \in S\backslash X_i}a_j w^{X_i}_{I\union\set{J}}
- \sum_{j \in X_i}a_j w^{X_i}_{I}.
$$
By Lemma~\ref{lem:w}, $w^{X_i}_{I\union\set{j}} = w^{X_i}_{I}$ for $j \in X_i$ and
$w^{X_i}_{I\union\set{j}} = 0$ for $j \in S\backslash X_i$. The claim follows.
\end{proof}

%% file: la-upperbound.tex
In this section, we use Theorem~\ref{thm:la-decomp} to prove that  for the $\Knapsack$ problem the gap
of $\La{t}{K}$ approaches $1$ quickly as $t$ grows, where $K$ is 
the LP relaxation of (\ref{eq:knapsack-lp}).
First, we show that there is a set $S$ such that every feasible solution in $ \La{t}{K}$ 
satisfies the condition of the Theorem.

Given an instance $(C,V)$ of  $\Knapsack$, Let $OPT(C,V)$ denote the value of the optimal 
integral solution.

\begin{lemma}
  \label{lem:knapsack-bigobjects}
  Consider an instance $(C,V)$ of $\Knapsack$ and its linear programming relaxation $K$ given by (\ref{eq:knapsack-lp}). Let $t>1$ and $y \in \La{t}{K}$. Let $k<t$ and $S = \set{i \in V\st{v_i > {OPT(C,V)}/{k}}}$.
  Then:
\begin{equation*}  \sum_{i\in I \intersect S} c_i >C \implies  y_I = 0.
\end{equation*}
\end{lemma}

\begin{proof}
  There are three cases depending on the size of $I$:
  \begin{enumerate}
  \item $|I| \leq t-1$.
 Recall the capacity constraint $g(x) = C - \sum_{i \in V}c_ix_i\geq 0$.
  On the one hand, since $M_{\powerset[t-1]{V}}(g* y) \isPSD$, the diagonal  entry  $(g* y)_I$ must be 
  non-negative. On the other hand, writing out the definition of $(g*y)_I$ and noting that the coefficients $c_i$ are all non-negative, we infer
 $    (g* y)_I \leq Cy_I - \left(\sum_{i \in I}c_i\right)y_I$.
But by assumption, $\sum_{i \in I}c_i > C$. Thus we must have $y_I = 0$.

  \item $t \leq |I| \leq 2t-2$.
  Write $I=I_1 \union I_2 = I$ with $|I_1|, |I_2| \leq t-1$ and 
  $|I_1 \intersect S| \geq k$.
  Then $y_{I_1} = 0$.
  Since $M_{\powerset[t]{y}}\isPSD$, its 2-by-2 principal 
  submatrix $M_\set{I_1, I_2}(y) $ must also be positive semi-definite.
   \begin{equation*}
    M_\set{I_1, I_2}(y) =
     \begin{pmatrix}
        0 & y_I \\
        y_I & y_{I_2}
      \end{pmatrix} ,
  \end{equation*}
and it is easy to check that we must then have $y_I = 0$.
  
  \item $2t-1 \leq |I| \leq 2t$. Write $I=I_1 \union I_2 = I$ with $|I_1|, |I_2| \leq t$ and 
  $|I_1 \intersect S| \geq k$. Then $y_{I_1} = 0$ since $t \leq 2t-2$
  for all $t \geq 2$.
  By the same argument as in the previous case, we must then have $y_I = 0$.
  \end{enumerate}
\end{proof}

The following theorem shows that the integrality gap of the $t^{th}$ level
of the Lasserre hierarchy for $\Knapsack$ reduces quickly when $t$ 
increases. 

\begin{theorem}
 Consider an instance $(C,V)$ of $\Knapsack$ and its  LP relaxation $K$ given by (\ref{eq:knapsack-lp}). Let $t\geq 2$.
 Then  
 $$\text{Value}( \La{t}{K})\leq (1+\frac{1}{t-1}) OPT,$$
 and so the integrality gap at the $t$-th level of the Lasserre hierarchy is at most $1+1/(t-1)$.
\end{theorem}
\begin{proof}
Let $S = \set{i \in V\st{v_i > {OPT(C,V)}/{(t-1)}}}$. Let $y \in \La{t}{K}$. 
If $|I\intersect S|\geq t-1$, then the elements of $I\cap S$ have total value greater than $OPT(C,V)$, so they must not be able to fit in the knapsack: their total capacity exceeds $C$, and so by 
Lemma~\ref{lem:knapsack-bigobjects} we have $y_I=0$. Thus  the condition of Theorem~\ref{thm:la-decomp} holds for $k=t-1$.

Therefore,
  $\project{y}{\powerset[2]{V}}$ is a convex combination of $w^{X_i}$ with $X_i\subseteq S$, thus $\text{Value}(y)\leq \max_i \text{Value}(w^{X_i})$. 
By the first and third properties of the Theorem, we have:
$$\text{Value}(w^{X_i})\leq \sum_{j\in X_i} v_j + \text{Value}(\La{1}{K_i}).$$
By the nesting property of the Lasserre hierarchy, Lemma~\ref{ref:known}, and definition of $S$,
$$\text{Value}(\La{1}{K_i})\leq \text{Value}(K_i)\leq OPT(C-\text{Cost}(X_i),V\setminus S))+OPT(C,V)/(t-1).$$
By the second property of the Theorem, $w^{X_i}$ is in  $ \La{t-k}{K}\subseteq K$, so it must satisfy the capacity constraint, so 
$\sum_{i\in X_i}c_i \leq \sum_{i\in I } c_i \leq C$, so $X_i$ is feasible. Thus:
\begin{eqnarray*}
\text{Value}(y)&\leq & \max_{\text{feasible }X\subseteq S} \left(
		\sum_{j\in X} v_j + OPT(C-\text{Cost}(X),V\setminus S)) \right) +OPT(C,V)/(t-1)\\
\end{eqnarray*}
The first expression in the right hand side is equal to $OPT(C,V)$, hence the Theorem.
\end{proof}

%% file: conclusion.tex
We have shown that for $\Knapsack$, an integrality gap of $2-\epsilon$ persists
up to a linear number of rounds in the Sherali-Adams hierarchy.
This broadens the class of problems for which 
Sherali-Adams is not strong enough to capture the instrinsic difficulity
of problems.

On the other hand, our positive result for Lasserre opens the posibility that
lower bounds in the Lasserre hierarchy good indicators of the intrinsic
dificulty of the problem, thus
encourages more investigation on the effect of the hierarchy on 
``easy'' problems ($\name{Spanning Tree}$, $\name{Bin Packing}$, etc.)

One obstacle along this line is the fact that the second positive semidefinite 
constraint of the hierarchy ($M_{\powerset{t}{V}}({g_\ell * y}) \isPSD$)
is notoriously hard to deal with, especially when $g_\ell$ contains 
many variables (in the lowerbounds for $k$-CSPs~\cite{Schoen,T}, the authors 
are able to get around this by constructing vectors for only valid assignments, 
an approach that is possible only when all the constraints are ``small''.)
Clearly, both lower bounds and upper bounds for the Lasserre hierarchy for 
problems with large constraints remain interesting to pursue.

%% file: ack.tex
Clare Mathieu would like to thank Eden Chlamtac for stimulating discussions.

%% file: appendix.tex
\section{Full proof of Theorem~\ref{thm:sa-lowerbound}}
\label{appendix:SAproof}

\begin{proof}
Let $t\geq 2$.
Consider the instance $K$ of $\Knapsack$ with $n$ objects where
$c_i = v_i = 1$ for all $i \in V$ and capacity $C = 2(1-\epsilon)$. 
Let $\alpha = {C}/({n+(t-1)(1-\epsilon)})$ and consider the vector $y \in [0,1]^{\powerset[t]{V}}$ defined by
\begin{align*}
	 \left\{
					\begin{aligned}
						 y_{\emptyset}=1 & \\
						 y_{\{ i \} }=\alpha &\\
						 y_I=0 & \ \ \text{if}\ |I| > 1\\
					\end{aligned}
				\right.
\end{align*}
We claim that $y \in \SA{t}{K}$.	Consider any subset $U \subseteq V$ such that $|U| \leq t$. We have
	\begin{align*}
		M_{\powerset{U}}(y) = 
			\begin{pmatrix}
				M_{\powerset[1]{U}}(y) & 0\\
				0 & 0
			\end{pmatrix},~~~~~\text{with}~~~~
		M_{\powerset[1]{U}}(y) =
			\begin{pmatrix}
				1 		 & \alpha & \alpha & \cdots & \alpha\\
				\alpha & \alpha & 0			 & \cdots & 0\\
				\alpha & 0			& \alpha & \cdots & 0\\
				\vdots & \vdots & \vdots & \ddots & \vdots\\
				\alpha & 0			& 0			 & \cdots & \alpha
			\end{pmatrix}.
	\end{align*}
	Since $|U| \leq t < n$, $|U|\alpha \leq 1$, and  it is easy to see that this implies $M_{\powerset[1]{U}}(y) \isPSD$, and so $M_{\powerset{U}}(y) \isPSD$.
	
	Next, let $g(x) = C - \sum_{i \in V}c_ix_i$ and 
	consider any subset $W \subseteq V$ such that $|W| \leq t-1$.
	Again, we have
	\begin{align*}
		M_{\powerset{W}}(g*y) = 
			\begin{pmatrix}
				M_{\powerset[1]{W}}(g*y) & 0\\
				0 & 0
			\end{pmatrix},~
					M_{\powerset[1]{W}}(g*y) =
			\begin{pmatrix}
				C - n\alpha	& (C-1)\alpha & (C-1)\alpha & \cdots & (C-1)\alpha\\
				(C-1)\alpha & (C-1)\alpha & 0			 			& \cdots & 0\\
				(C-1)\alpha & 0						& (C-1)\alpha & \cdots & 0\\
				\vdots 			& \vdots 			& \vdots 			& \ddots & \vdots\\
				(C-1)\alpha & 0						& 0			 			& \cdots & (C-1)\alpha
			\end{pmatrix}.
	\end{align*}
	Since $|W| \leq t-1$, by definition of $\alpha$ we have $|W|(C-1)\alpha \leq C-n\alpha$, and it is easy to see that this implies $M_{\powerset[1]{W}}(g*y) \isPSD$, and so $M_{\powerset{W}}(g*y) \isPSD$. Thus $y \in \SA{t}{K}$. 
	
The  integer optimum has value  1, so the integrality gap is at least the value of $y$, which is $n\alpha=2(1-\epsilon)/(1+(t-1)(1-2\epsilon)/n)$. The supremum over all $\epsilon$ is $2/(1+(t-1)/n)$, and the supremum of that over all $n$ is 2, so the integrality gap is at least 2.

On the other hand, it is well-known that the base linear program $K$ has value at most $2 OPT$ (that is an immediate consequence of Lemma~\ref{ref:known}), hence, by the nesting property, every linear program in the hierarchy has integrality gap exactly equal to 2.
\end{proof}

\section{A linear formulation of the Sherali-Adams hierarchy}
\label{a:sa}
For any constraint $g_\ell(x) \geq 0$ in the definition of the base polytope
$K$ and any subsets 
$I, J \subseteq V$, the following constraint is a consequence of the
fact that $x \in [0,1]^n$:
\begin{align}
	g_\ell(x)\prod_{i \in I}x_i\prod_{j \in J}(1-x_j) \geq 0.
	\label{eq:sa-prod}
\end{align}
If $x$ is indeed integral, then $x_i^k = x_i$ for any $k \geq 1$.
Thus, the constraint obtained by expanding~\eqref{eq:sa-prod} and replacing
$x_i^k$ by $x_i$ holds in $P$ and can be added to strengthen the 
relaxation.
However, this constraint is not linear.
To preserve the linearity of the system, 
each product $\prod_{i \in I}x_i$ is replaced by a variable $y_I$.

In addition, to keep the number of variables from growing exponentially, we restrict 
ourselves to only variables $y_I$ such that $|I| \le t$.
By this, we ``lift'' the polytope $K$ to a polytope 
$\SA{t}{K} \subseteq [0,1]^{\powerset[t]{V}}$.
 
\begin{definition}
  \label{def:SA-linear}
Let $K$ be a polytope defined as in equation~\ref{eq:defK}.
  For any $1 \leq t \leq n$, the \emph{$t$-th Sherali-Adams lifted polytope}
  $\SA{t}{K}$ is
  defined by
  $$
    \SA{t}{K} = \set{
      y \in \powerset[t]{[n]} \st{ 
			y_\emptyset = 1,\ \text{and}\
      g'_{\ell,I,J}(y) \geq 0\ \text{for any}\ \ell\ \text{and}\ 
      I,J \subseteq V\ \text{s.t.}\ |I \union J| \leq t-1 }}
  $$
  where $g'_{\ell,I,J}(y)$ is obtained by:
  \begin{enumerate}
    \item multiplying $g_\ell(x)$ by $\prod_{i\in I}x_i\prod_{j\in J}(1-x_j)$;
    \item expanding the result and replacing each $x_i^k$ by $x_i$; and
    \item replacing each $\prod_{i \in S}x_i$ by $y_S$.
  \end{enumerate}

  We say that a point $x \in [0,1]^V$ belongs to the \emph{$t$-th Sherali-Adams
  polytope $\sa{t}{K}$} iff there exists a $y \in \SA{t}{K}$ such that
  $y_\set{i} = x_i$ for all $i \in V$.
\end{definition}

In particular, in the case of $\Knapsack$, $\SA{t}{K}$ is the set of all 
points in $[0,1]^{\powerset[t]{V}}$ that satisfy the following constraints
for any $I, J \subseteq V$ such that $I \intersect J = \emptyset$ and
$|I| + |J| \leq t-1$:
\begin{equation}
	\sum_{i = 1}^n c_i\sum_{L \subseteq J}(-1)^{|L|}y_{I\union L\union\set{i}}
		\leq C\sum_{L \subseteq J}(-1)^{|L|}y_{I\union L}, 
	\label{eq:knapsack-prod}
\end{equation}
and
$$	0 \leq \sum_{L \subseteq J}(-1)^{|L|}y_{I\union L\union\set{i}}
		\leq \sum_{L \subseteq J}(-1)^{|L|}y_{I\union L},~~ \forall i \in V.$$

For a proof that this definition is equivalent to 
Definition~\ref{def:SA-psd}, we refer the reader to Laurent's paper~\cite{L01}.